\documentclass[12pt]{article}
\usepackage[usenames]{color}
\usepackage[ukrainian,english]{babel}
\usepackage[cp1251]{inputenc}
\usepackage{doc}
\usepackage{amsmath}
\usepackage{amsfonts}   
\usepackage{cite} 
\usepackage{fullpage}                           
\usepackage[lmargin=0.6in,rmargin=0.6in,tmargin=0.6in,headsep=.2in]{geometry}
\usepackage{graphicx}
\usepackage{forloop}
\usepackage{etoolbox}
\usepackage{calc}
\usepackage{wrapfig,lipsum,booktabs} 
\usepackage{xtab} 
\usepackage[dvipsnames]{xcolor}
\usepackage[normalem]{ulem}
\usepackage{sectsty}
\usepackage{amsmath,amsfonts,amssymb,amsthm,epsfig,epstopdf,url,array}
 \usepackage[retainorgcmds]{IEEEtrantools}
 \usepackage{makeidx,epsfig,lscape}
 \usepackage{xcolor,pict2e}
\usepackage{upgreek}
   
\makeatother
\usepackage{makeidx}
\makeindex
\makeatletter
\renewcommand{\@biblabel}[1]{#1.}

\newtheorem{te}{Theorem}           
\newtheorem{de}{Definition}          
\newtheorem{ce}{Consequence}    
\newtheorem{prope}{Proposition} 
\newtheorem{lemma}{Lemma}   

\begin{document}
\centerline{ \bf\large{ Tax systems for sustainable economic development.}}

\vskip 5mm

{\bf \centerline {\Large N.S. Gonchar \footnote{ This work is partially supported by the Fundamental Research Program of the Department of Physics and Astronomy of the National Academy of Sciences of Ukraine "Building and researching financial market models using the methods of nonlinear statistical physics and the physics of nonlinear phenomena N 0123U100362." }} }

\vskip 5mm
\centerline{\bf { Bogolyubov Institute for Theoretical Physics }} 
\centerline{\bf {of the National Academy of Sciences of Ukraine.}}
\vskip 2mm

 \begin{abstract}
 A complete description of taxation systems that ensure sustainable economic development is given. These tax systems depend on production technologies and gross output volumes. Explicit formulas for such dependencies are found.
  In a sustainable economy, the value added either exceeds or is strictly less than the value of the product produced. The latter is determined by the tax system.  
 
The concept of perfect taxation systems is introduced and their explicit form is found.
  For perfect taxation systems, it is proved that the vector of gross output should belong to the interior of the cone formed by the vectors of the columns of the total cost matrix.
 It is shown that under perfect taxation systems the vector of gross output must satisfy a certain system of linear homogeneous equations.
  It is shown, that under certain conditions there are tax systems under which certain industries require subsidies for their existence.  Under such taxation systems, the industries that require subsidies are identified.
 
 The family of all non negative solutions of the system of linear equations and inequalities is constructed, which allowed us to formulate a criterion for describing all equilibrium states in which partial clearing of  markets occurs.
\end{abstract}

\centerline{{\bf Keywords:} Technological mapping; Economic balance; }
\centerline{ Clearing markets; Vector of taxation;}
 \centerline{Sustainable economic development.}
  \centerline{Aggregated economy description.}
 \section{Introduction.}

 Sustainable development is defined as the production of goods and services that leads to the complete clearing of markets in a given period of economic functioning under the influence of the market mechanism of pricing goods and services.
 The market mechanism for establishing the prices of goods and services is understood as the equilibrium of supply and demand for resources and produced goods in the production process, taking into account the tax system. 
 For taxation systems that ensure sustainable economic development, the nonlinear system of equations with respect to the equilibrium price vector is transformed into a linear system of equations with a strictly positive solution, which always exists under minor assumptions from the economic point of view.

 The main problem of the work was to find out whether such a market mechanism exists for the main production models. The study of this problem was started in the papers \cite{11Gonchar11},  \cite{12Gonchar12},  \cite{10Gonchar2},  \cite{11Gonchar2}. 
 From a mathematical point of view, this is a question of whether there is an equilibrium price vector and taxation system under which there is equality of supply and demand, determined by the structure of production technologies, production volumes and demand for goods and services produced in the production process by consumers.
 This is the main problem solved in the first part of this paper.  It describes all taxation systems that ensure sustainable economic development in the input-output production model.
 
 For such taxation systems, an equilibrium price vector is constructed at which demand equals supply. This equilibrium vector depends on a taxation system that ensures sustainable development. The existence of such an equilibrium price vector is a confirmation of the fact that a market mechanism for pricing goods and services exists, which is important from an economic point of view.
 
  Among the taxation systems that ensure sustainable development, perfect taxation systems are distinguished. They are characterized by the fact that the gross value added in an industry is equal to the value of the product created in the same industry.

 We describe taxation systems under which the economic system is able to function in a  mode with subsidies. That is, there exists a market pricing mechanism in which there is a strictly positive equilibrium price vector such that under this taxation system certain industries require subsidies to exist.
 
 In real economic systems, the taxation system may not coincide with the taxation system that ensures  sustainable development. In this case, there is no complete clearing of the markets for an equilibrium price vector.
 The second part of the work is devoted to the description of all equilibrium states, i.e., those states in which only a partial clearing of the markets takes place.
Such states of equilibrium describe the possible overproduction of goods and services in the economic system.

The work gives a complete description of the equilibrium states under which only partial clearing of the markets takes place. To do this, the problem of a complete description of non negative solutions of a linear system of equations and inequalities is first solved.
 
 Then the Theorem on the necessary and sufficient conditions that each solution of a linear system of equations and inequalities corresponds to an equilibrium price vector is established.
 The method of constructing such an equilibrium state under which overproduction in the economic system will be the smallest is indicated.
 
 It reduces to constructing a solution to a linear system of equations and inequalities, which is a solution to a certain problem of quadratic programming.
 
 Clarifying the conditions of sustainable development of the economy is extremely important to avoid undesirable development scenarios.  The definition of what we mean by sustainable development at  the micro-economic level is presented in a number of works \cite{11Gonchar11},  \cite{12Gonchar12},  \cite{10Gonchar2},  \cite{11Gonchar2}. 
At the macroeconomic level, this problem was formulated in the articles \cite{6Gonchar} - \cite{8Gonchar}. 
 In  Definition 1, we determine the tax systems that ensure sustainable development under an equilibrium price vector determined by the equality of supply and demand for production resources and produced goods.

 Theorem \ref{Nato1} is an auxiliary statement that provides sufficient conditions for the existence of a solution with respect to the price vector and which has proved to be extremely important from the point of view of the concept of sustainable economic development.
 
 Theorem \ref{Kur3} provides the necessary and sufficient  conditions under which a tax system ensures  sustainable development of the economy.
 In fact, it provides a description of tax systems under which sustainable economic development  takes place.
 
 Theorem \ref{PTPPP4} formulates necessary and sufficient conditions for taxation systems under which sustainable development of the economy takes place. It states that there are no taxation systems other than those described   that ensure  sustainable economic development.

Theorem \ref{1KissP1} establishes the conditions for the gross output vector under which the economic system is able to function in the mode of sustainable development.

 Theorem \ref{KissP1} states that under the taxation system given by the formula (\ref{Kur4}), the set of all industries is divided into two non-intersecting sets in which the added value either does not exceed the value of the  product produced or strictly exceeds it.
 
  In the Definition \ref{kissTk1} we determine the industries of the economic system that require subsidies under a competitive equilibrium price vector.
 
 Theorem \ref{PPT1} gives sufficient conditions for tax systems  under which certain industries require subsidies.
 
Theorem \ref{nessuf1}  establishes the existence of  perfect tax systems.

 Theorem \ref{Nato8}  gives the sufficient conditions under which, with a perfect taxation system, the economic system is able to function in the mode of sustainable development.
 
More specifically,  provided that the vector of gross outputs is a solution of the system of equations in the "input-output" production model with a positive right-hand side, there is always a perfect tax system given by an explicit formula under which the economic system is able to function in the mode of sustainable development. Moreover, the gross value added  generated  in each industry is strictly positive. It is shown that in the mode of sustainable development under the established taxation system, the gross value added  created is equal to the value of the product created in this industry. 

 In real economic systems, not all branches of production have a strictly positive gross value added.  There may be several reasons for this: obsolete production technologies, significant imports of consumer goods, raw materials, etc.

 Not all tax systems ensure the sustainable development of the economy. The equilibrium price vector generated by the tax system and production technologies and the demand for resources may lead to the negative value added  in certain industries.

 Theorem \ref{10Nato8} establishes the existence of an equilibrium price vector in  this case and gives expressions for the amount of subsidies in those industries in which values added created   are negative.

 As a result of the study of the conditions of sustainable development of the economy it was established that the taxation system in the mode of sustainable development depends on production technologies and volumes of gross outputs. It is the correct choice of the tax system that will create a competitive equilibrium price vector that will ensure sustainable development.
 
 In real economic systems, taxation consists of direct taxes on production and indirect taxes on consumers, which also affect production.
  As a result, the final taxation may, as a rule, differ from the taxation of the mode  sustainable development  in which markets are completely cleared.
 
 Due to the taxation that has developed in the real economic system, there will be a partial clearing of the markets and therefore a part of the produced goods will not find their consumers.
 
Section 3 examines exactly this case.
Theorem \ref{TVYA} describes all non negative solutions of the system of linear equations and inequalities, which is important for describing all equilibrium states with excess supply.
 
 Proposition \ref{g1} proves that the minimum distance to the vector of the right hand side of the linear system of equations and inequalities is reached on the set of all non-negative solutions of the linear system of equations and inequalities.
 This minimum is global. The latter allows to find this minimum by solving some quadratic programming problem.

 Theorems \ref{TtsyVtsja5} and \ref{2TtsyVtsja4} give a sufficient condition for the existence of a solution of the system of equations with respect to the price vector, provided that the right-hand side of this system of equations belongs to the cone formed by the column vectors of the non-negative matrix.

 Based on the theorems \ref{TtsyVtsja5} and \ref{2TtsyVtsja4}, the definition \ref{kolja2} is given, in which each non-negative solution of the linear system of equations and inequalities corresponds to an equilibrium price vector.
  Theorem \ref{myktina19} is the basis for the description of all equilibrium states in which partial clearing of the markets takes place.
 
 Theorem \ref{MykHon1} gives the necessary and sufficient conditions for the existence of an equilibrium price vector under which partial clearing of the markets takes place. Namely, every non-negative solution of the system of linear inequalities and equations corresponds to the equilibrium price vector.
 
 Definition \ref{myktinavitka1} describes the equilibrium state that has the least excess supply. To find it, one should solve a certain problem of quadratic programming to find the solution of a linear system of equations and inequalities, and based on this solution, construct the real consumption vector, then find the equilibrium price vector and calculate the level of excess supply.

 \section{Description of taxation systems.}
In this section, we describe taxation systems in the economy model described by "input - output" production technologies.
 In the model of  economy, described by "input - output" technology, the   matrix 
 $A=||a_{ki}||_{k,i=1}^n$ is supposed a non negative, productive and  indecomposable one.
 Further, we assume that the gross output vector $x=\{x_k\}_{k=1}^n$ satisfies the system of equations
  \begin{eqnarray}\label{Kur7} 
  x_k=\sum\limits_{i=1}^na_{ki} x_i +c_k+e_k-i_k,\quad k=\overline{1,n},
  \end{eqnarray}  
where $c=\{c_k\}_{k=1}^n,$ $e=\{e_k\}_{k=1}^n,$ $i=\{i_k\}_{k=1}^n$ are the vectors of internal consumption, export and import, correspondingly.
We also  assume that equilibrium price vector $p=\{p_k\}_{k=1}^n$ satisfies the set of equations 
  \begin{eqnarray}\label{Kur8} 
  p_k=\sum\limits_{i=1}^n a_{sk} p_s +\delta_k,\quad k=\overline{1,n}.
  \end{eqnarray}   
Under these conditions,   the solutions of the set of equations (\ref{Kur7}) and  (\ref{Kur8})  always exist. The non negative vector 
 $\delta=\{\delta_k\}_{k=1}^n$ we call the vector of added values.
In the market economy, the values of added values $\delta_k, \ k=\overline{1,n},$ are unknown. 

Prices for goods and services are formed under the influence of market forces of supply and demand. This is the essence of the market economy.  The pricing system is also affected by the taxation system. How do the created  added values depend on the type of taxation system?

Is there a principle of formation of prices for goods and services that takes into account the taxation system and market pricing mechanisms?

The basis of this principle are formulated below.\\
1) In each period of the economy functioning, there are resources to ensure production with given technologies.\\
2) Aggregate demand for resources and produced goods, which is determined by production technologies, must be equal to the aggregate supply of resources and produced goods.\\

We  call these two principles the principles of sustainable economic development.

Mathematically, they was formulated in papers
\cite{11Gonchar11},  \cite{12Gonchar12},  \cite{10Gonchar2},  \cite{11Gonchar2}. 
This formulation became possible thanks to a new approach to the description of the economic systems presented in the monograph \cite{Gonchar2}

We summarize it in the following  Definition \ref{PTPPP1} (see also \cite{11Gonchar11},  \cite{12Gonchar12},  \cite{10Gonchar2},  \cite{11Gonchar2}).
\begin{de}\label{PTPPP1}
The economic system functions in the mode of sustainable development under the taxation system
$\pi=\{\pi_i\}_{i=1}^n$ if there exists a strictly positive solution relative to the vector $p_0=\{p_i^0\}_{i=1}^n$ of the system of equations
\begin{eqnarray}\label{PTPPP2}
  \sum\limits_{i=1}^n a_{ki}\frac{(1-\pi_i)x_i p_i^0}{ \sum\limits_{s=1}^na_{si}p_s^0 }=(1-\pi_k) x_k, \quad  i=\overline{1,n},
\end{eqnarray}
 satisfying conditions
 \begin{eqnarray}\label{PTPPP3}
p_k^0-\sum\limits_{s=1}^np_s^0 a_{s k}>0,  \quad k=\overline{1,n},
\end{eqnarray}  
where $A=||a_{ki}||_{k,i=1}^n$ is a non-negative non-decomposable  productive matrix, $x=\{x_k\}_{k=1}^n$ is a strictly positive solution to the set of equations (\ref{Kur7}).
 \end{de}
 
 The following Definition \ref{Kur1} is important for describing taxation systems that ensure  sustainable development of economy.
\begin{de}\label{Kur1}
The taxation system $\pi=\{\pi_k\}_{k=1}^n, \  0<\pi_k<1, \ k=\overline{1,n},$ ensures the sustainable development of the economy under a strictly positive gross output vector $x=\{x_k\}_{k=1}^n$, if there exists a strictly positive equilibrium price vector
$p_0=\{p_k^0\}_{k=1}^n$,
solving the system of equations
\begin{eqnarray}\label{Kur2}
  \sum\limits_{i=1}^n a_{ki}\frac{(1-\pi_i)x_i p_i^0}{ \sum\limits_{s=1}^na_{si}p_s^0 }=(1-\pi_k) x_k, \quad  i=\overline{1,n},
\end{eqnarray}  
and such that 
\begin{eqnarray}\label{111Kur1}
p_k^0-\sum\limits_{s=1}^np_s^0 a_{s k}>0,  \quad k=\overline{1,n},
\end{eqnarray}  
where $A=||a_{ki}||_{k,i=1}^n$ is a non-decomposable non-negative productive matrix.
\end{de}

The next Theorem \ref{Nato1} is an auxiliary result which will help to describe taxation systems that ensure sustainable development of the economy.
\begin{te}\label{Nato1}
Let $z=\{z_i\}_{i=1}^n$ be a strictly positive vector from $R_+^n$ and let  $A=||a_{ki}||_{k,i=1}^n$ be a non-negative  non-decomposable matrix. The set of equations
\begin{eqnarray}\label{Nato2}
\frac{z_k}{\sum\limits_{i=1}^n a_{ki}z_i}=\frac{p_k}{\sum\limits_{s=1}^n a_{sk}p_s}, \quad k=\overline{1,n},
\end{eqnarray}
has a strictly positive solution $p=\{p_i\}_{i=1}^n \in R_+^n.$ 
\end{te}
\begin{proof}
Let us consider the nonlinear set of equations
\begin{eqnarray}\label{Nato3}
\frac{p_k+y_k \sum\limits_{s=1}^n a_{sk}p_s}{1+\sum\limits_{k=1}^n y_k \sum\limits_{s=1}^n a_{sk}p_s}=p_k, \quad k=\overline{1,n},
\end{eqnarray}
 where we denoted by
 $$y_k=\frac{z_k}{\sum\limits_{i=1}^n a_{ki}z_i}, \quad k=\overline{1,n}. $$
\end{proof}
This set of equations has a solution in the set  $P=\{p=\{p_i\}_{i=1}^n \in R_+^n, \  \sum\limits_{i=1}^n p_i=1\},$ since the left part of this set of equations is a map that maps the set $P$ into itself and is continuous on it. Really, due to Brouwer Theorem  \cite{Nirenberg} there exists a solution $p_0=\{p_i^0\}_{i=1}^n$ of the set of equations (\ref{Nato3}) in the set $P.$
From the set of equations (\ref{Nato3}) it follows that $p_0=\{p_i^0\}_{i=1}^n$
is a solution of the set of equations
\begin{eqnarray}\label{Nato4}
y_k \sum\limits_{s=1}^n a_{sk}p_s^0=\lambda p_k^0,\quad k=\overline{1,n},
\end{eqnarray}
where $\lambda=\sum\limits_{k=1}^n y_k \sum\limits_{s=1}^n a_{sk}p_s^0.$
We  prove that $\lambda>0$ and the solution $p_0=\{p_i^0\}_{i=1}^n$ of the system of equations (\ref{Nato4})  is strictly positive due to the indecomposability of the matrix $A$.
Indeed, the vector $p_0=\{p_i^0\}_{i=1}^n$ satisfies the system of equations
(\ref{Nato4}), which can be written in operator form
\begin{eqnarray}\label{phuph3}
V^T p_0=\lambda p_0,
\end{eqnarray}
or
\begin{eqnarray}\label{phuph3}
[V^T]^{n-1} p_0=\lambda^{n-1} p_0,
\end{eqnarray}
where we denoted the matrix $V=||v_{ij} ||_{ij=1}^n,$  $v_{ij}=a_{ij}y_j, \ i, j=\overline{1,n}.$
Due to the fact that the vector $p_0$ belongs to the set $P$ and the matrix $V$ is non-negative and indecomposable, the vector $[V^T]^{n-1} p_0$ is strictly positive. It follows that $\lambda>0$ and the vector $p_0$ is strictly positive.

 Let us prove that 
 $\lambda=1.$ Multiplying  by $ p_k^0 $  the left and right hand sides of the equality
$$ \frac{z_k}{y_k}=\sum\limits_{i=1}^n a_{ki}z_i   $$
 and summing up over $k$ from $1$ to $n$ we obtain
 $$  \sum\limits_{k=1}^n\frac{p_k^0z_k}{y_k}=\sum\limits_{i=1}^n \sum\limits_{k=1}^n a_{ki}p_k^0 z_i =\lambda \sum\limits_{i=1}^n\frac{p_i^0z_i}{y_i}.    $$
 Since  $\sum\limits_{k=1}^n\frac{p_k^0z_k}{y_k}>0$ we have $ \lambda =1.$
 Theorem \ref{Nato1} is proved.
 
 \begin{ce}\label{TinTinPPE1}
 The solution of the set of equations (\ref{Nato2}) constructed in Theorem \ref{Nato1} is determined uniquely with accuracy up to a positive constant.
  \end{ce}
 
 \begin{te}\label{Kur3} 
 Let $x=\{x_i\}_{i=1}^n$ be a strictly positive gross output vector from $R_+^n$ and let  $A=||a_{ki}||_{k,i=1}^n$ be a non-negative  non-decomposable productive matrix.
Necessary and sufficient conditions
 for the taxation system $\pi=\{\pi_k\}_{k=1}^n,\  0<\pi_k<1, \ k=\overline{1,n},$ to ensure sustainable development of the economy is  the  following  representation of it
 \begin{eqnarray}\label{Kur4}
 \pi_i=1-b \frac{\sum\limits_{j=1}^n a_{ij}z_j}{x_i}, \quad 0<b <\min\limits_{1\leq i\leq n} \frac{x_i}{\sum\limits_{j=1}^n a_{ij}z_j},\quad  i=\overline{1,n},
 \end{eqnarray}
for a certain strictly positive vector $z=\{z_i\}_{i=1}^n \in R_+^n, $ satisfying conditions
 \begin{eqnarray}\label{111Kur2}
 \frac{z_i}{\sum\limits_{j=1}^n a_{ij}z_j}>1, \quad 
 i=\overline{1,n}.
 \end{eqnarray}
 \end{te}
 \begin{proof}
 Necessity. Let a  taxation system $\pi=\{\pi_k\}_{k=1}^n,\  0<\pi_k<1, \ k=\overline{1,n},$  ensure sustainable development of the economy. Then the equalities (\ref{Kur2}) are true. From them it follows that 
  \begin{eqnarray}\label{1Kur4}
 (1-\pi_k)x_k=\sum\limits_{i=1}^n a_{ki}z_i^0,\quad  k=\overline{1,n},
 \end{eqnarray}
  where
 \begin{eqnarray}\label{2Kur4}
   z_i^0=\frac{(1-\pi_i)x_i p_i^0} {\sum\limits_{s=1}^n a_{si}p_s^0}>0,\quad  i=\overline{1,n},   
 \end{eqnarray}   
$ p_0=\{p_i^0\}_{i=1}^n $ is a strictly positive solution to the set of equations (\ref{Kur2}).
Substituting (\ref{1Kur4}) into (\ref{2Kur4}) we obtain that $ p_0=\{p_i^0\}_{i=1}^n $ solves the set of equations
 \begin{eqnarray}\label{3Kur4}
\frac{p_k^0}{\sum\limits_{s=1}^n a_{sk}p_s^0}=\frac{z_k^0}{\sum\limits_{i=1}^n a_{ki}z_i^0}, \quad k=\overline{1,n}.
 \end{eqnarray} 
 If to substitute instead of
 $$\frac{p_k^0}{\sum\limits_{s=1}^n a_{sk}p_s^0}, \quad  k=\overline{1,n}, $$ 
 the right side of equality (\ref{3Kur4})
 into (\ref{Kur2}) we get the set of equations relative to the taxation system $\pi=\{\pi_k\}_{k=1}^n,\  0<\pi_k<1, \ k=\overline{1,n},$ 
 \begin{eqnarray}\label{4Kur4}
  \sum\limits_{i=1}^n a_{ki}\frac{(1-\pi_i)x_i z_i^0}{ \sum\limits_{s=1}^n a_{ij}z_j^0 }=(1-\pi_k) x_k, \quad  i=\overline{1,n},
\end{eqnarray} 
The solution of the set of equations (\ref {4Kur4}) is 
 \begin{eqnarray}\label{5Kur4}
 \pi_k=1- b \frac{\sum\limits_{i=1}^n a_{ki}z_i^0}{x_k},\quad  \quad 0<b <\min\limits_{1\leq k\leq n} \frac{x_k}{\sum\limits_{j=1}^n a_{kj}z_j^0},\quad  k=\overline{1,n}.
 \end{eqnarray} 
 Due to the fact that equilibrium price vector
$p_0=\{p_i^0\}_{i=1}^n$  also satisfies to the set of equations (\ref{3Kur4}) it follows that the inequalities (\ref{111Kur2}) are true, since the equilibrium price vector $p_0=\{p_i^0\}_{i=1}^n$ satisfies the inequalities (\ref{111Kur1}).

 Sufficiency. Suppose that the taxation system is given by the formula (\ref{Kur4}) for which the inequalities (\ref{111Kur2}) are true. Let us prove the existence of strictly positive solution $p_0=\{p_i^0\}_{i=1}^n$  of the set of equations (\ref{Kur2}) satisfying the inequalities  (\ref{111Kur1}). Substituting (\ref{Kur4}) into (\ref{Kur2}),  we obtain the set of equations relative the equilibrium price vector $p_0=\{p_i^0\}_{i=1}^n$
\begin{eqnarray}\label{Kur5}
  \sum\limits_{i=1}^n a_{ki}\frac{b\sum\limits_{j=1}^n a_{ij}z_j p_i^0}{ \sum\limits_{s=1}^na_{si}p_s^0 }=b\sum\limits_{j=1}^n a_{kj}z_j, \quad  i=\overline{1,n}.
\end{eqnarray}  
If to choose the vector $p_0=\{p_i^0\}_{i=1}^n$ so that it would satisfy the system of equations
 \begin{eqnarray}\label{Kur6}
 \frac{p_i^0}{\sum\limits_{s=1}^n a_{si}p_s^0}=\frac{z_i}{\sum\limits_{j=1}^n a_{ij}z_j},\quad  i=\overline{1,n},
 \end{eqnarray}  
 then it will satisfy the set of equations (\ref{Kur2}). But the strictly positive  solution to the set of equations (\ref{Kur6}) always exists, due to Theorem \ref{Nato1}. Due to the inequalities (\ref{111Kur2}), the inequalities  (\ref{111Kur1}) are true.
 \end{proof}
 
  Next Theorem 
 \ref{PTPPP4} states that there are no taxation systems other than those described above that ensure sustainable economic development.

  \begin{te}\label{PTPPP4}
   Let $x=\{x_i\}_{i=1}^n$ be a strictly positive gross output vector from $R_+^n$ and let  $A=||a_{ki}||_{k,i=1}^n$ be a non-negative  non-decomposable productive matrix.
   An economy functions in the mode of sustainable  development if and only if  for the taxation system  $\pi=\{\pi_k\}_{k=1}^n,\  0<\pi_k<1, \ k=\overline{1,n},$ the representation  (\ref{Kur4}) is true, which satisfies the conditions  (\ref{111Kur2}).
\end{te}
\begin{proof}
Necessity.
Let an economy functions in the mode of sustainable development. Then, due to Definition \ref{PTPPP1}, there exist a taxation system $\pi=\{\pi_k\}_{k=1}^n,\  0<\pi_k<1, \ k=\overline{1,n},$   a strictly positive equilibrium price vector
  $ p_0=\{p_i^0\}_{i=1}^n,$
solving the set of equations (\ref{PTPPP2}), and  satisfying the conditions (\ref{PTPPP3}). Let us prove that for the taxation system
 $\pi=\{\pi_k\}_{k=1}^n,\  0<\pi_k<1, \ k=\overline{1,n},$  
the representation  (\ref{Kur4}) is valid, and it satisfies the conditions
  (\ref{111Kur2}).  Introduce the denotations  $X_i=x_i p_i^0,\ i=\overline{1,n},$
  $\bar a_{ij}=\frac{p_i^0 a_{ij}}{p_j^0},\ i,j=\overline{1,n}.$  Relative to the vector $\pi=\{\pi_k\}_{k=1}^n,\  0<\pi_k<1, \ k=\overline{1,n},$ we obtain the set of equations
\begin{eqnarray}\label{1PTPPP4}
  \sum\limits_{i=1}^n \bar a_{ki}\frac{(1-\pi_i)X_i }{ \sum\limits_{s=1}^n \bar a_{si} }=(1-\pi_k) X_k, \quad  i=\overline{1,n}.
\end{eqnarray} 
Let us denote
\begin{eqnarray}\label{PTPPP5}
Y=\{Y_i\}_{i=1}^n, \quad
\frac{(1-\pi_i)X_i }{ \sum\limits_{s=1}^n \bar a_{si} }=Y_i, \quad  i=\overline{1,n}.
\end{eqnarray} 
Then the vector $Y=\{Y_i\}_{i=1}^n$ solves the set of equations
  \begin{eqnarray}\label{PTPPP6}
  \sum\limits_{i=1}^n \bar a_{ki}Y_i=\sum\limits_{s=1}^n \bar a_{sk} Y_k, \quad  i=\overline{1,n}.
\end{eqnarray}
Due to Lemma 3 \cite{11Gonchar11}, there exists the strictly positive solution to the set of equations (\ref{PTPPP6}). We choose it such that $\sum\limits_{i=1}^n Y_i=1$. Then, the others strictly positive  solutions can be represented in the form
$c_0Y,$  where $c_0>0$ and is arbitrary.
Then, we have 
\begin{eqnarray}\label{PTPPP7}
  1-\pi_i=\frac{c_0Y_i\sum\limits_{s=1}^n \bar a_{si} }{X_i}, \quad  i=\overline{1,n}.
\end{eqnarray}
We choose $c_0$ such that $0<\pi_i<1, \ i=\overline{1,n}.$
Then, $0<c_0< \min\limits_{1\leq i \leq n}\frac{X_i}{Y_i\sum\limits_{s=1}^n \bar a_{si}}.$
Let us show that the representation (\ref{Kur4}) is true and the conditions
  (\ref{111Kur2}) are valid.
   If to put $z_i=\frac{Y_i}{p_i^0}, \ i=\overline{1,n},$  and to take into account (\ref{PTPPP6}), then we get
\begin{eqnarray}\label{PTPPP8}
  1-\pi_i=\frac{c_0\sum\limits_{j=1}^n  a_{ij}z_j }{x_i}, \quad  i=\overline{1,n}.
\end{eqnarray}
Since $1<\pi_i<1, \ i=\overline{1,n},$ we have
$0<c_0<\min\limits_{1\leq i \leq n}\frac{x_i}{\sum\limits_{j=1}^n  a_{i j}z_j}.$

 The proof of the fact that the inequalities  (\ref{111Kur2}) takes place is such that as the proof of sufficiency in Theorem \ref{Kur3}.
Really, since we proved the representation   (\ref{Kur4}) for the  taxation system, then, as in the proof of sufficiency of Theorem \ref{Kur3}, we obtain that the strictly positive equilibrium price vector satisfies to the set of equations
 \begin{eqnarray}\label{PTPPP9}
 \frac{p_i^0}{\sum\limits_{s=1}^n a_{si}p_s^0}=\frac{z_i}{\sum\limits_{j=1}^n a_{ij}z_j},\quad  i=\overline{1,n}.
 \end{eqnarray} 
 Since the economy functions in the mode of sustainable  development, the inequalities (\ref{PTPPP3}) are true, which proves the needed.
 
 Sufficiency. Suppose that taxation system is given by the formula  (\ref{Kur4}) which satisfy the conditions (\ref{111Kur2}). Then, as in the proof of the sufficiency of Theorem \ref{Kur3}, we obtain that the equilibrium price vector is a solution to the set of equations (\ref{PTPPP9}) the solution of which exists due to Theorem \ref{Nato1}. It also solves the set of equations (\ref{PTPPP2})  and satisfies the condition (\ref{PTPPP3}), since the inequalities 
  \begin{eqnarray}\label{PTPPP10}
 \frac{p_i^0}{\sum\limits_{s=1}^n a_{si}p_s^0}=\frac{z_i}{\sum\limits_{j=1}^n a_{ij}z_j}>1,\quad  i=\overline{1,n},
 \end{eqnarray} 
 are valid.
 Theorem \ref{PTPPP4} is proved.
\end{proof}

  The following Theorem is a repetition of Theorem 16 from \cite{11Gonchar11}. Here we only specify what the taxation system is in this case, which was not in Theorem 16. 
\begin{te}\label{1KissP1}
 Let $x=\{x_i\}_{i=1}^n$ be a strictly positive gross output vector from $R_+^n,$ such that the vector $(1-\pi)x=\{1-\pi_k)x_k\}_{k=1}^n$ belongs to the interior of the cone created by the column vectors  of the matrix $A(E-A)^{-1},$ where   $A=||a_{ki}||_{k,i=1}^n$ is a non-negative non-decomposable  productive matrix. For the taxation system, given by the formula (\ref{Kur4}), the economy system, described by  "input - output" production model, can function in the mode of sustainable development.
\end{te} 
 \begin{proof}
 Due to Theorem 16 from \cite{11Gonchar11}, the conditions of which are fulfilled, the economy system can function in the mode of sustainable development. Theorem
 \ref{1KissP1} indicates only what taxation system is. 
 
Below, we give a new proof of this Theorem.  From the taxation system, given by the formula (\ref{Kur4}), we have
 \begin{eqnarray}\label{1KissP3} 
 (1-\pi_i)x_i= b\sum\limits_{i=1}^na_{ki}z_i^0=\sum\limits_{i=1}^na_{ki} z_i, \quad k=\overline{1,n},
  \end{eqnarray}
 for a certain strictly positive  vector $z_0=\{z_i^0\}_{i=1}^n,$
where we put $z=\{z_i\}_{i=1}^n = b \{z_i^0\}_{i=1}^n.$
Thanks to Theorem \ref{PTPPP4},  the equilibrium price vector  $p_0=\{p_i^0\}_{i=1}^n$ solves the set of equations
 \begin{eqnarray}\label{1KissP2}
 \frac{p_i^0}{\sum\limits_{s=1}^n a_{si}p_s^0}=\frac{z_i}{\sum\limits_{j=1}^n a_{ij}z_j},\quad  i=\overline{1,n}.
 \end{eqnarray}  
 Due to the conditions of  Theorem \ref{1KissP1},  for the vector $ (1-\pi)x$  the representation
  \begin{eqnarray}\label{1KissP4} 
  (1-\pi)x=A(E-A)^{-1}\alpha,  \quad \alpha=\{\alpha_k\}_{k=1}^n, \quad \alpha_k >0,\quad 
 k=\overline{1,n},
  \end{eqnarray}  
is true.  The representation (\ref{1KissP4})  for the vector $z=\{z_i\}_{i=1}^n$ gives us
the formula $z=(E-A)^{-1}\alpha$ for a  certain strictly positive vector $\alpha=\{\alpha_i\}_{i=1}^n.$
From this it follows that  for the vector   $z=\{z_i\}_{i=1}^n$  the inequalities 
 \begin{eqnarray}\label{1KissP5}
 \frac{p_i^0}{\sum\limits_{s=1}^n a_{si}p_s^0}=\frac{z_i}{\sum\limits_{j=1}^n a_{ij}z_j}>1,\quad  i=\overline{1,n},
 \end{eqnarray} 
 are true. 
 This proves Theorem \ref{1KissP1}.
 \end{proof}

 \begin{te}\label{KissP1} 
  Let $x=\{x_i\}_{i=1}^n$ be a strictly positive gross output vector from $R_+^n$ and let  $A=||a_{ki}||_{k,i=1}^n$ be a non-negative  non-decomposable productive matrix.
  Suppose that the economy be in the mode of sustainable development with a strictly positive gross output vector $x=\{x_i\}_{i=1}^n$, solving the system of equations (\ref{Kur7}), under the taxation system
(\ref{Kur4}).  Then the set of pure branches $N=\{1,2, \ldots,n\}$ is divided into two disjoint sets $I$ and $J$ such that $I\cup J=N$ and the inequalities
 \begin{eqnarray}\label{KissP2}
 \Delta_k \leq C_k+E_k-I_k, \quad k \in I, \quad  \Delta_k > C_k+E_k-I_k, \quad k \in J, 
  \end{eqnarray} 
  are  valid, where 
$$   \Delta_k=X_k\left(1-\sum\limits_{s=1}^n \frac{p_s^0 a_{s i}}{p_i^0}\right), \quad X_k=p_k^0 x_k , $$ 
 \begin{eqnarray}\label{KissP3}   
  C_k=c_k p_k^0,    \quad E_k=e_k p_k^0, \quad I_k=i_k p_k^0,  \quad k=\overline{1,n},
\end{eqnarray} 
$p_0=\{p_i^0\}_{i=1}^n$ is an equilibrium price vector.
 \end{te}
 \begin{proof}
 Provided that the economy  is in the mode of sustainable development, the equilibrium price vector   $ p_0=\{p_i^0\}_{i=1}^n$ satisfies the system of equations
 \begin{eqnarray}\label{KissP4}
 \frac{p_i^0}{\sum\limits_{s=1}^n a_{si}p_s^0}=\frac{z_i}{\sum\limits_{j=1}^n a_{ij}z_j}>1,\quad  i=\overline{1,n},
 \end{eqnarray}   
for a strictly positive vector $z= \{z_i\}^n$ such that 
 \begin{eqnarray}\label{KissP5}
(1-\pi_k)x_k=b\sum\limits_{j=1}^n a_{kj}z_j, \quad k=\overline{1,n},
  \end{eqnarray} 
  where 
  \begin{eqnarray}\label{KissP6}
 \pi_k=1- b \frac{\sum\limits_{i=1}^n a_{ki}z_i}{x_k},\quad  \quad 0<b <\min\limits_{1\leq k\leq n} \frac{x_k}{\sum\limits_{j=1}^n a_{kj}z_j},\quad  k=\overline{1,n}.
 \end{eqnarray} 
Two cases are possible
    \begin{eqnarray}\label{KissP7}
  \frac{z_i}{\sum\limits_{j=1}^n a_{ij}z_j} \leq  \frac{x_i}{\sum\limits_{j=1}^n a_{ij}x_j},\quad  i\in I,
 \end{eqnarray}   
 and
 \begin{eqnarray}\label{KissP8}
  \frac{z_i}{\sum\limits_{j=1}^n a_{ij}z_j} >  \frac{x_i}{\sum\limits_{j=1}^n a_{ij}x_j},\quad  i \in J.
 \end{eqnarray}   
 
  In the first case  (\ref{KissP7}) we have 
  \begin{eqnarray}\label{KissP9} 
  \frac{p_i^0}{\sum\limits_{s=1}^n a_{si}p_s^0} \leq \frac{x_i}{\sum\limits_{j=1}^n a_{ij}x_j},\quad  i \in I,
   \end{eqnarray} 
 or 
  \begin{eqnarray}\label{KissP10} 
  \frac{\sum\limits_{s=1}^n a_{si}p_s^0}{p_i^0} \geq \frac{\sum\limits_{j=1}^n a_{ij}x_j}{x_i},\quad  i \in I.
   \end{eqnarray}  
    From here we get
      \begin{eqnarray}\label{KissP11} 
1-  \frac{\sum\limits_{s=1}^n a_{si}p_s^0}{p_i^0} \leq 1-\frac{\sum\limits_{j=1}^n a_{ij}x_j}{x_i}=\frac{c_i+e_i-i_i}{x_i},\quad  i \in I.
   \end{eqnarray}    
 From the last inequality we have  
      \begin{eqnarray}\label{KissP12} 
\Delta_i=p_i^0x_i\left(1-  \frac{\sum\limits_{s=1}^n a_{si}p_s^0}{p_i^0}\right) \leq p_i^0( x_i-\sum\limits_{j=1}^n a_{ij}x_j)=C_i+E_i-I_i,\quad  i \in I.
   \end{eqnarray}  
    The consideration  of the second case (\ref{KissP8})  proceeds similarly. As a result, we get that 
      \begin{eqnarray}\label{KissP12} 
\Delta_i=p_i^0x_i\left(1-  \frac{\sum\limits_{s=1}^n a_{si}p_s^0}{p_i^0}\right) > p_i^0( x_i-\sum\limits_{j=1}^n a_{ij}x_j)=C_i+E_i-I_i,\quad  i \in J.
   \end{eqnarray}
    Theorem \ref{KissP1} is proved.  
 \end{proof}
 
 \begin{de}\label{kissTk1}
 We say  say that under a taxation system 
 $\pi=\{\pi_i\}_{i=1}^n $ certain industries need subsidies if the strictly positive equilibrium price  vector $p^0=\{p^0_i\}_{i=1}^n $ solving the set of equations 
 \begin{eqnarray}\label{kissTk2}
  \sum\limits_{i=1}^n a_{ki}\frac{(1-\pi_i)x_i p_i^0}{ \sum\limits_{s=1}^na_{si}p_s^0 }=(1-\pi_k) x_k, \quad  i=\overline{1,n},
 \end{eqnarray}
 is such that the added values created in these industries are negative for
 the strictly positive gross output vector $x=\{x_k\}_{k=1}^n,$   satisfying the set of equations (\ref{Kur7}). 
 \end{de}
 
  \begin{te}\label{PPT1}
Let the gross output vector   $x=\{x_i\}_{i=1}^n$  be a solution to the set of equations (\ref{Kur7}) and let  $A=||a_{k i}||_{k,i=1}^n$ be a non negative non-decomposable productive matrix.  
 For the vector of taxation $\pi=\{ \pi_i\}_{i=1}^n,$ such that 
 \begin{eqnarray}\label{PPT2}
 1-\pi_i=b\frac{\sum\limits_{j=1}^n a_{i j}z_j}{x_i}, \quad  i=\overline{1,n},\quad 0< b<\min\limits_{1\leq i\leq  n}\frac{x_i}{\sum\limits_{j=1}^n a_{i j}z_j},
 \end{eqnarray}
and non empty set
 \begin{eqnarray}\label{PPT5}
 J=\left\{k, \  \frac{z_k}{\sum\limits_{j=1}^na_{kj}z_j}<1\right\} \subset N,
  \end{eqnarray} 
there exists a strictly positive equilibrium price vector  $p_0=\{p_i^0\}_{i=1}^n$ solving the set of equations
\begin{eqnarray}\label{PPT3}
  \sum\limits_{i=1}^n a_{ki}\frac{(1-\pi_i)x_i p_i^0}{ \sum\limits_{s=1}^na_{si}p_s^0 }=(1-\pi_k) x_k, \quad  i=\overline{1,n},
 \end{eqnarray}
and such that  the industries, the indexes of which belongs to the set $J,$ needs subsidies.
The subsides into  the $k$-th industry should  not be  smaller than 
  \begin{eqnarray}\label{PPT6}
x_k p_k^0\left(\frac{\sum\limits_{j=1}^na_{kj}z_j}{z_j}- 1\right),\quad k \in J.
 \end{eqnarray}
 \end{te}
 \begin{proof}
Let us show that the set $J$ does not coincide with the set $N=\{1,2, \ldots,n\}.$
 
The proof from the opposite. If the set $J$ would coincide with the set $N,$ then we get the set of inequalities
   \begin{eqnarray}\label{PPPT7}
   z_k <\sum\limits_{j=1}^na_{kj}z_j, \quad k=\overline{1,n}.
    \end{eqnarray}
 Due to  productivity of  the matrix $A,$ the only bounded solution of  set of inequalities
(\ref{PPPT7} ) is zero solution. Contradiction, since the vector  $z=\{z_i\}_{i=1}^n$ is strictly positive one.

 Since the vector $p_0=\{p_i^0\}_{i=1}^n$ is a strictly positive solution of the set of equations (\ref{Nato2}), then we have
 \begin{eqnarray}\label{PPT7}
 p_i^0- \sum\limits_{s=1}^na_{si} p_s^0=p_i^0\left(1-\frac{\sum\limits_{j=1}^na_{ij}z_j}{z_i}\right),\quad i=\overline{1,n}.
 \end{eqnarray}
 From the formula (\ref{PPT7}), we obtain the 
 formula  
 \begin{eqnarray}\label{PPT8}
\Delta_i =x_i( p_i^0- \sum\limits_{s=1}^na_{si} p_s^0)=x_i p_i^0\left(1-\frac{\sum\limits_{j=1}^na_{ij}z_j}{z_i}\right)<0 , \quad  i \in J.
\end{eqnarray}
Theorem \ref{PPT1} is proved. 
 \end{proof}

  Below, we will study taxation systems of a special type, which we call perfect.

\begin{de}\label{Kur9}
A system of taxation $\pi=\{\pi_k\}_{k=1}^n,\  0<\pi_k<1, \ k=\overline{1,n},$  that ensures sustainable economic development   we call perfect if the equalities 
  \begin{eqnarray}\label{Kur10} 
 x_k \delta_k=y_k p_k, 
 \quad k=\overline{1,n},
  \end{eqnarray}   
  are true, where the gross output vector $x=\{x_k\}_{k=1}^n$ satisfies the system of equations (\ref{Kur7}) and the equilibrium price vector satisfies the set of equations (\ref{Kur8}). 
\end{de} 
 \begin{te}\label{nessuf1}
 Let  $A=||a_{k i}||_{k,i=1}^n$ be a non negative non-decomposable   productive matrix and let the  vector $x=\{x_k\}_{k=1}^n$  be a strictly positive solution to the set of equations (\ref{Kur7}). Then, there exists  always the perfect  system of taxation $\pi=\{\pi_k\}_{k=1}^n$ given by the formula
 \begin{eqnarray}\label{1Kur7}
 1-\pi_i=b\frac{\sum\limits_{j=1}^n a_{i j}x_j}{x_i}, \quad  i=\overline{1,n},\quad 0< b<\min\limits_{1\leq i\leq  n}\frac{x_i}{\sum\limits_{j=1}^n a_{i j}x_j},
 \end{eqnarray}
  and  a strictly positive equilibrium price vector  $p_0=\{p_i^0\}_{i=1}^n$ solving the set of equations
\begin{eqnarray}\label{nessuf2}
  \sum\limits_{i=1}^n a_{ki}\frac{(1-\pi_i)x_i p_i^0}{ \sum\limits_{s=1}^na_{si}p_s^0 }=(1-\pi_k) x_k, \quad  i=\overline{1,n},
 \end{eqnarray}
which also satisfies the set of equation 
 \begin{eqnarray}\label{nessuf3}
 p_i^0- \sum\limits_{s=1}^na_{si} p_s^0=p_i^0\left(1-\frac{\sum\limits_{j=1}^na_{ij}x_j}{x_i}\right),\quad i=\overline{1,n}.
 \end{eqnarray}
\end{te}
\begin{proof}
To prove Theorem \ref{nessuf1}, it needs to indicate  the added values $\delta_i, i=\overline{1,n},$ taxation system $\pi=\{\pi_k\}_{k=1}^n$ under which the equalities (\ref{Kur10}) are true. 
Since the gross output vector $x=\{x_k\}_{k=1}^n$ satisfies the system of equations (\ref{Kur7}) and the equilibrium price vector $p_0=\{p_k^0\}_{k=1}^n$ satisfies the set of equations (\ref{Kur8}),
the equalities (\ref{Kur10}) are equivalent to the equalities
 \begin{eqnarray}\label{Kur11}
p_k^0\sum\limits_{i=1}^n a_{ki}x_i=\sum\limits_{s=1}^n a_{sk}p_s^0 x_k,\quad k=\overline{1,n}.
 \end{eqnarray}
 So, from the set of equations (\ref{Kur11}) it follows that  the equilibrium price vector  $ p_0=\{p_i^0\} _{i=1}^n$ should satisfy the set of equations
  \begin{eqnarray}\label{1Kur11}
  \frac{p_k^0}{\sum\limits_{s=1}^n a_{sk}p_s^0}=\frac{x_k}{\sum\limits_{i=1}^n a_{ki}x_i},\quad k=\overline{1,n},
  \end{eqnarray}
 the strictly positive solution of which
 $ p_0=\{p_i^0\} _{i=1}^n$ there exists always due to Theorem \ref{Nato1}.
 Substituting the solution $ p_0=\{p_i^0\} _{i=1}^n$ of the set of equations (\ref{1Kur11})  into the set of equations (\ref{nessuf2}) we get the set of equations relative to the vector $\pi=\{\pi_i\}_{i=1}^n$
 \begin{eqnarray}\label{nessuf5}
  \sum\limits_{i=1}^n a_{ki}\frac{(1-\pi_i) x_i^2}{ \sum\limits_{j=1}^na_{ij}x_j }=(1-\pi_k) x_k, \quad  i=\overline{1,n}.
 \end{eqnarray}
  If we put that,  
 \begin{eqnarray}\label{nessuf4}
 1-\pi_i=b\frac{\sum\limits_{j=1}^n a_{i j}x_j}{x_i}, \quad  i=\overline{1,n},\quad 0< b<\min\limits_{1\leq i\leq  n}\frac{x_i}{\sum\limits_{j=1}^n a_{i j}x_j},
 \end{eqnarray}
then the system of equations (\ref{nessuf5}) is satisfied.
\end{proof}

 Below,  we consider the economy system described by "input - output" model with non negative non-decomposable productive matrix of direct costs 
$A=||a_{k i}||_{k,i=1}^n$ and strictly positive  gross output vector $x=\{x_k\}_{k=1}^n$  that satisfies the set of equations
\begin{eqnarray}\label{Nato9}
x_k=\sum\limits_{i=1}^n a_{ki} x_i +c_k+e_k -i_k, \quad k=\overline{1,n},
 \end{eqnarray}
with the following limitations
\begin{eqnarray}\label{real1}
 c_k+e_k -i_k>0,  \quad  k=\overline{1,n},  
  \end{eqnarray}
where
$$ c=\{c_k\}_{k=1}^n, \quad e=\{e_k\}_{k=1}^n, \quad  i=\{i_k\}_{k=1}^n, $$
are the vectors of final consumption, export and import, correspondingly.

 \begin{te}\label{Nato8}
  Let  $A=||a_{k i}||_{k,i=1}^n$ be a non negative non-decomposable   productive matrix and let the strictly positive  gross output vector   $x=\{x_i\}_{i=1}^n$  be a solution to the set of equations (\ref{Nato9}) with limitations (\ref{real1}).  For the vector of taxation $\pi=\{ \pi_i\}_{i=1}^n,$ where 
 \begin{eqnarray}\label{Nato5}
  \pi_i=1-b\frac{ \sum\limits_{k=1}^na_{ik}x_k}{x_i}, \quad i=\overline{1,n}, \quad  0 <b<\min\limits_{1\leq i\leq n} \frac{x_i} { \sum\limits_{k=1}^na_{ik}x_k},
 \end{eqnarray}
there exists a strictly positive equilibrium price vector  $p_0=\{p_i^0\}_{i=1}^n$ solving the set of equations
\begin{eqnarray}\label{Nato6}
  \sum\limits_{i=1}^n a_{ki}\frac{(1-\pi_i)x_i p_i^0}{ \sum\limits_{s=1}^na_{si}p_s^0 }=(1-\pi_k) x_k, \quad  i=\overline{1,n},
 \end{eqnarray}
which also satisfies the set of equation 
 \begin{eqnarray}\label{Nat13}
 p_i^0- \sum\limits_{s=1}^na_{si} p_s^0=p_i^0\left(1-\frac{\sum\limits_{j=1}^na_{ij}x_j}{x_i}\right)>0,\quad i=\overline{1,n}.
 \end{eqnarray}
and is such that the economy system described by "input - output" model with non-decomposable productive matrix of direct costs 
$A=||a_{k i}||_{k,i=1}^n$ and gross output vector $x=\{x_k\}_{k=1}^n$   is capable to function in the mode of sustainable development.
 \end{te}
\begin{proof}
The proof of the  first part of Theorem \ref{Nato8} follows from Theorem \ref{nessuf1}. 
Since the vector  $x=\{x_i\}_{i=1}^n$  is a solution to the set of equations (\ref{Nato9}) with limitations (\ref{real1}), we obtain that   
$$ \frac{ x_k}{\sum\limits_{k=1}^na_{k i}x_i}=1+\frac{c_k+e_k -i_k}{\sum\limits_{i=1}^na_{k i}x_i}>1, \quad  k=\overline{1,n}. $$
It is not difficult to find that the
created  added value  is given by the formula
$$\delta_i^0= p_i^0- \sum\limits_{s=1}^na_{si} p_s^0=$$
\begin{eqnarray}\label{Nato13}
\left(\frac{x_k}{\sum\limits_{j=1}^na_{ij}x_j}-1\right)\frac{\sum\limits_{s=1}^na_{si} p_s^0}{p_i^0}p_i^0=p_i^0\left(1-\frac{\sum\limits_{j=1}^na_{ij}x_j}{x_i}\right)>0,\quad  i=\overline{1,n}.
 \end{eqnarray}
 Theorem \ref{Nato8} is proved.
 \end{proof}
  \begin{ce}\label{Nat16}
 In the mode of sustainable development the gross added value created in the $i$-th industry is equal to the value of the final product created in this industry.
 \end{ce}
 \begin{proof}
 From the formula (\ref{Nato13}) we have
  \begin{eqnarray}\label{Nat14}
\frac{\Delta_i}{X_i}=\frac{x_i\delta_i^0}{x_ip_i^0}=1-\frac{\sum\limits_{j=1}^n \frac{p_i^0 a_{ij}}{p_j^0}p_j^0x_j}{p_i^0x_i}=\frac{C_i+E_i-I_i}{X_i},\quad i=\overline{1,n},
\end{eqnarray}
where we introduced the denotations $X_k=x_k p_k^0,$ $C_k=c_k p_k^0,$ $E_k=e_k p_k^0,$ $I_k=i_k p_k^0.$ From here we get $\Delta_i=C_i+E_i-I_i.$
Consequence \ref{Nat16} is proved.
 \end{proof}
Based on Theorems \ref{Nato1} - \ref{Nato8}, we conclude that the deviation  the value $(C+E -I)- \Delta$ from zero depends on the taxation system. If the taxation system is given by formulas (\ref{Nato5}), then the perfect sustainable economic development takes place. Theorem \ref{KissP1} states that this deviation depends on the deviation of the taxation system from the perfect one.  
So, the characteristic of the taxation system in the mode of sustainable development is the number of negative and positive signs of  of the value $(C+E -I)- \Delta.$ 
If the sign of $(C+E -I)- \Delta$  is negative, then we are dealing with high-tech production with high added value. Conversely, if the sign of $(C+E -I)- \Delta$  is positive, then we are dealing with low-tech production with low added value.

The number of negative and positive signs of the value $(C+E -I)- \Delta$ we call the signature of taxation.

Every economy in the world is open to its environment. That is, they all exchange goods, labor resources and capital among themselves. This happens due to uneven distribution of resources, excess production of goods. Some countries are rich in resources, while others have high-tech industries. Because of this, some import raw materials, while others import goods with high added value.
Below we define the conditions under which an open economic system that imports both produced goods and resources can operate in the mode with subsides of certain industries.

 Theorem \ref{10Nato8} takes into account such a situation.
Below,  we consider the economy system described by "input - output" model with non negative non-decomposable productive matrix of direct costs 
$A=||a_{k i}||_{k,i=1}^n$ and gross output vector $x=\{x_k\}_{k=1}^n$  that satisfies the set of equations
\begin{eqnarray}\label{10Nato9}
x_k=\sum\limits_{i=1}^n a_{ki} x_i +c_k+e_k -i_k, \quad k=\overline{1,n},
 \end{eqnarray}
with the limitations
 \begin{eqnarray}\label{subsid1}
  c_k+e_k -i_k \geq 0,\   k \in I \neq \emptyset, \quad  c_k+e_k -i_k< 0, \ k \in J, \ J\neq \emptyset,
 \end{eqnarray}  
where  $c=\{c_k\}_{k=1}^n,$ $e=\{e_k\}_{k=1}^n,$ $i=\{i_k\}_{k=1}^n,$ are the vectors of final consumption, export and import, correspondingly.

 \begin{te}\label{10Nato8}
  Let  $A=||a_{k i}||_{k,i=1}^n$ be a non negative non-decomposable   productive matrix and let
the strictly positive  gross output vector   $x=\{x_i\}_{i=1}^n$  be a solution to the set of equations (\ref{10Nato9}) with limitations (\ref{subsid1}).  For the vector of taxation $\pi=\{ \pi_i\}_{i=1}^n,$ where 
 \begin{eqnarray}\label{10Nato5}
 1-\pi_i=b\frac{\sum\limits_{j=1}^n a_{i j}x_j}{x_i}, \quad  i=\overline{1,n},\quad 0< b<\min\limits_{1\leq i\leq  n}\frac{x_i}{\sum\limits_{j=1}^n a_{i j}x_j},
 \end{eqnarray}
there exists a strictly positive equilibrium price vector  $p_0=\{p_i^0\}_{i=1}^n$ solving the set of equations
\begin{eqnarray}\label{10Nato6}
  \sum\limits_{i=1}^n a_{ki}\frac{(1-\pi_i)x_i p_i^0}{ \sum\limits_{s=1}^na_{si}p_s^0 }=(1-\pi_k) x_k, \quad  i=\overline{1,n},
 \end{eqnarray}
and  such that the economy system described by "input - output" model with non negative  non-decomposable productive matrix of direct costs 
$A=||a_{k i}||_{k,i=1}^n$ and gross output vector $x=\{x_k\}_{k=1}^n$   is capable to function in the mode with subsides. The subsides into  the $k$-th industry should  not be  smaller than 
  \begin{eqnarray}\label{subsid3}
x_k p_k^0\left(\frac{\sum\limits_{j=1}^na_{kj}x_j}{x_k}- 1\right),\quad k \in J.
 \end{eqnarray}
 \end{te}
 \begin{proof}
Due to Theorem \ref{10Nato8} conditions, there exists strictly positive solution to the set of equations (\ref{10Nato6}) that
satisfies to the set of equations
\begin{eqnarray}\label{10Nato7}
  \frac{ p_i^0}{ \sum\limits_{s=1}^na_{si}p_s^0 }=\frac{ x_i}{\sum\limits_{k=1}^na_{ik}x_k}, \quad  i=\overline{1,n}.
 \end{eqnarray}  
Since the strictly positive vector  $x=\{x_i\}_{i=1}^n$  is a solution to the set of equations (\ref{10Nato9}) we obtain that   
$$ \frac{ x_k}{\sum\limits_{k=1}^na_{k i}x_i}=1+\frac{c_k+e_k -i_k}{\sum\limits_{i=1}^na_{k i}x_i}\geq 1, \quad  k\in I, $$
$$ \frac{ x_k}{\sum\limits_{k=1}^na_{k i}x_i}=1+\frac{c_k+e_k -i_k}{\sum\limits_{i=1}^na_{k i}x_i}<1, \quad  k \in J. $$

 So, the vector $p_0=\{p_i^0\}_{i=1}^n$ is a strictly positive solution of the set of equations (\ref{Nato2}), then we have
 \begin{eqnarray}\label{10Nato13}
 p_i^0- \sum\limits_{s=1}^na_{si} p_s^0=p_i^0\left(1-\frac{\sum\limits_{j=1}^na_{ij}x_j}{x_i}\right)<0,\quad i \in J.
 \end{eqnarray}
From the formula (\ref{10Nato13}) it follows that subsides in the $k$-th industry should not be smaller than
   \begin{eqnarray}\label{subsid4}
x_k p_k^0\left(\frac{\sum\limits_{j=1}^na_{kj}x_j}{x_k}- 1\right),\quad k \in J.
 \end{eqnarray}
 Theorem \ref{10Nato8} is proved.
 \end{proof}

 \begin{ce}\label{me1}  In the mode of sustainable  economic development, the   gross output vector $X=\{X_k\}_{k=1}^n$ in value indicators satisfies the system of equations  
  \begin{eqnarray}\label{me2}
  \sum\limits_{i=1}^n \bar a_{ki} X_i
= \sum\limits_{s=1}^n \bar a_{sk} X_k, \quad  k=\overline{1,n},
 \end{eqnarray}
 where we introduced the denotations $X_k=p_k^0x_k, \ \bar a_{k i}=
 \frac{p_k^0a_{ki}}{p_i^0}, k,i =\overline{1,n}. $ The equilibrium price  vector $p_0=\{p_i^0\}_{i=1}^n $ solves the set of equations (\ref{10Nato7}) and the gross output vector $x=\{x_i\}_{i=1}^n $ solves the set of equations (\ref{10Nato9}).  In this case  the following formulae 
  \begin{eqnarray}\label{9me2}
C_k+  E_k- I_k=(1-\sum\limits_{s=1}^n \bar a_{sk} )X_k, \quad  k=\overline{1,n},
 \end{eqnarray}
 take place.
 \end{ce}
 \begin{proof}
 The set of equations  (\ref{me2}) is a direct consequence of the set of equations  (\ref{10Nato7}).
 Due to Lemma 3 (see \cite{11Gonchar11}), there  exists always the solution to the set of equations (\ref{me2}) relative to the vector $X=\{X_i\}_{i=1}^n. $  From the fact that 
  \begin{eqnarray}\label{9me3}
C_k+  E_k - I_k=X_k-   \sum\limits_{i=1}^n \bar a_{ki} X_i
= X_k(1-\sum\limits_{s=1}^n \bar a_{sk}), \quad  k=\overline{1,n},
 \end{eqnarray} 
we get the needed statement. 
 \end{proof}
 \begin{ce}\label{me3} 
  In the mode of sustainable economic development, the following formulas 
\begin{eqnarray}\label{me4} 
  \Delta_k=(1-\sum\limits_{s=1}^n \bar a_{sk})X_k, \quad
k=\overline{1,n}, 
 \end{eqnarray}  
 are true. The taxation vector $\pi=\{\pi_k\}_{k=1}^n $ is given by the formula
 \begin{eqnarray}\label{me5} 
 \pi_k=1-b\sum\limits_{s=1}^n \bar a_{sk}=1-b\left (1-\frac{\Delta_k}{X_k}\right), \quad  0< b< \min\limits_{1\leq k \leq n}\frac{1}{\sum\limits_{s=1}^n \bar a_{sk}},\quad  k=\overline{1,n}.
  \end{eqnarray} 
 \end{ce}
\begin{proof}
The proof of the formula (\ref{me5}) follows from the Theorem \ref{10Nato8}. Really, from the formula (\ref{10Nato5}) we have
 \begin{eqnarray}\label{me7}
 1-\pi_i=b\frac{\sum\limits_{j=1}^n a_{i j}x_j}{x_i}=b\frac{\sum\limits_{s=1}^n a_{si}p_s^0}{p_i^0}=b \sum\limits_{s=1}^n \bar a_{si}=b\left (1-\frac{\Delta_i}{X_i}\right),   \quad  i=\overline{1,n}.
 \end{eqnarray}
  \end{proof}

\section{Equilibrium states with partial market clearing.}
The previous section completely describes the taxation systems that ensure  the sustainable development of the economic system. According to Definition \ref{Kur1}, there is always a strictly positive equilibrium price vector $p_0=\{p_k^0\}_{k=1}^n$, satisfying the system of equations (\ref{Kur2}), i.e., there is a complete clearing of markets. But when the taxation system does not coincide with the described ones, then the equilibrium vector of prices $p_0=\{p_k^0\}_{k=1}^n$, must satisfy the system of equations and inequalities
\begin{eqnarray}\label{2Cur1}
  \sum\limits_{i=1}^n a_{ki}\frac{(1-\pi_i)x_i p_i^0}{ \sum\limits_{s=1}^na_{si}p_s^0 }\leq (1-\pi_k) x_k, \quad  i=\overline{1,n}.
\end{eqnarray}  
In this case, we say about only a partial clearing of the markets. Our task is to describe all equilibrium states in which only partial market clearing occurs. The latter means that all non-negative solutions of the system of equations and inequalities (\ref{2Cur1}) should be described.

Let us introduce the denotations
$(1-\pi_k)x_k=b_k, \ k=\overline{1,n}$ and a vector $b=\{b_k\}_{k=1}^n$. We assume that the vector $x=\{x_k\}_{k=1}^n$ is a strictly positive solution to the set of equations (\ref{Kur7}) and  tax system $\pi=\{\pi_k\}_{k=1}^n$ is such that $0<\pi_k<1,\ k=\overline{1,n}.$
From this assumptions we obtain that the vector $b=\{b_k\}_{k=1}^n$ is strictly positive.
So, we need to describe all non negative solutions $p=\{p_k\}_{k=1}^n$ of the nonlinear set of equations and inequalities
 $$  \sum\limits_{j=1}^n  a_{ij}\frac{b_j p_j}{\sum\limits_{s=1}^n a_{sj}p_s}=b_i \quad i \in I,     $$
\begin{eqnarray}\label{mytinq1}
   \sum\limits_{j=1}^n a_{ij}\frac{b_j p_j}{\sum\limits_{s=1}^n a_{sj}p_s}<b_i \quad i \in J,     
\end{eqnarray}
where $I$ and $J$ are non empty sets such that $I\cup J=N=\{1,2,\ldots,n\},$\ $I\cap J=\emptyset.$
\begin{lemma}\label{mytinq2}
Let $A=||a_{ij}||_{i,j =1}^n$ be a non negative  non zero  matrix and let $b=\{b_i\}_{=1}^n$ be a strictly positive vector. 
If $p=\{p_i\}_{i=1}^n$ is an  equilibrium   price vector which solves the system of equations and inequalities
 $$  \sum\limits_{j=1}^n  a_{ij}\frac{b_j p_j}{\sum\limits_{s=1}^n a_{sj}p_s}=b_i \quad i \in I,     $$
\begin{eqnarray}\label{mytinq3}
   \sum\limits_{j=1}^n a_{ij}\frac{b_j p_j}{\sum\limits_{s=1}^n a_{sj}p_s}<b_i \quad i \in J,     
\end{eqnarray}
in the set $P=\{p=\{p_i\}_{i=1}^n,\ p_i \geq 0,\ i=\overline{1,n}, \ \sum\limits_{i=1} ^n p_i=1\},$
then $p_i=0, \ i \in J.$ 
 \end{lemma}
\begin{proof}
Suppose that the vector $p=\{p_i\}_{i=1}^n$ is a solution of the system of equations and inequalities (\ref{mytinq3}), belonging to the set $P.$ Let us show that $p_i=0, \ i \in J.$ 
We  lead the proof from the opposite.
Let at least one component $p_k, k \in J,$ of the vector $p$ be strictly positive.
Then, multiplying by $p_i, i=\overline{1,n},$ the $i$-th equation or inequality and summing up the left and right parts, respectively, we obtain the inequality
$ \sum\limits_{j=1}^n b_j p_j< \sum\limits_{i=1}^n b_i p_i. $
Since  the vector $b$ is a strictly positive one, this inequality is impossible, because $ \sum\limits_{j=1}^n b_j p_j>0.$ Therefore, our assumption is not correct, and so $p_i=0, \ i \in J. $
\end{proof}

Suppose that $I$ be a nonempty  subset of
  indices $k \in N_0 =[1,2,\ldots,n]$. Let us consider the  system of equations 
\begin{eqnarray}\label{mykt17}
\sum\limits_{i=1}^l c_{ki} z_i =b_k, \quad k\in I,
\end{eqnarray}
and inequalities
\begin{eqnarray}\label{mykt18}
\sum\limits_{i=1}^l c_{ki} z_i  < b_k, \quad k\in J,
\end{eqnarray}
where $C=||c_{k i}||_{k=1,i=1}^{n,l}$
is a non negative non zero  matrix.
 First, to describe all the non-negative solutions of the set of equations and inequalities (\ref{mytinq3}),  we give the complete description of the non-negative  solutions of the system of equations and inequalities (\ref{mykt17}), (\ref{mykt18}).  

We denote by $c_i=\{c_{ki}\}_{k=1}^n, \ i=\overline{1,l},$ the $i$-th column of the matrix $C=||c_{k i}||_{k=1,i=1}^{n,l}.$
Let us consider the numbers $d_i=\min\limits_{1\leq k \leq n}\frac{b_k}{c_{ki}}, \ i=\overline{1,l}.$

Next Theorem \ref{TVYA} generalizes Theorem 9 from \cite{11Gonchar11}.
\begin{te}\label{TVYA} Let the strictly positive vector $b$ not belong to the cone formed by the column vectors $c_i=\{c_{ki}\}_{k=1}^n, \ i=\overline{ 1,l},$ of the non-negative  matrix $C$ such that $\sum\limits_{i=1}^l c_{ki}>0, k=\overline{1,n},  \sum\limits_{k=1}^n c_{ki}>0, i=\overline{1,l}.$  Any non-negative solution of the system of equations (\ref{mykt17}) and inequalities (\ref{mykt18}) is given by the formula
$$ z=\{c(\alpha)\alpha_i d_i\}_{i=1}^l,$$
where $\alpha=\{\alpha_i \}_{i=1}^l \in Q=
\{\alpha=\{\alpha_i \}_{i=1}^l, \ \alpha_i \geq 0, \sum\limits_{i=1}^l\alpha_i=1\},$
$$c(\alpha)=\min\limits_{1\leq k \leq n}
\frac{b_k}{[\sum\limits_{i=1}^l\alpha_i d_i c_i]_k}\geq 1.$$
The function $ c(\alpha)$ is bounded and continuous on the set $Q.$
\end{te}

\begin{proof}
Let $z_0=\{z_i^0\}_{i=1}^l$ be a certain vector that is a solution of the system of equations (\ref{mykt17}) and inequalities (\ref{mykt18}).
Let's denote
$$\alpha_i=\frac{\frac{z_i^0}{d_i}}{\sum\limits_{j=1}^l\frac{z_j^0}{d_j} },\quad i=\overline{1 ,l}.$$
Then, we have
$$ Cz_0=\sum\limits_{i=1}^l c_iz_i^0=\sum\limits_{j=1}^l\frac{z_j^0}{d_j} \sum\limits_{i=1}^l \alpha_i d_i c_i.$$
Because of
$$ \min\limits_{1\leq k \leq n} \frac{b_k}{[Cz_0]_k }=1,$$
we get
$$\sum\limits_{j=1}^l\frac{z_j^0}{d_j}= \min\limits_{1\leq k \leq n} \frac{b_k}{[\sum\limits_{i =1}^l \alpha_i d_i c_i]_k}.$$
It is obvious that, conversely, every vector $\alpha=\{\alpha_i \}_{i=1}^l \in Q$ corresponds to
solution of the system of equations (\ref{mykt17}) and inequalities (\ref{mykt18}), which is given by the formula
$$ z=\{a(\alpha)\alpha_i d_i\}_{i=1}^l,$$
where
$$a(\alpha)=\min\limits_{1\leq k \leq n}
\frac{b_k}{[\sum\limits_{i=1}^l\alpha_i d_i c_i]_k}.$$

Let us establish that $c(\alpha) \geq 1.$
It is obvious that $c_i d_i\leq b.$ Multiplying by $\delta_i\geq 0$ the left and right parts of the last inequality and summing over $i$ and assuming that $\sum\limits_{i=1}^l\delta_i >0$ we will get
$$\frac{\sum\limits_{i=1}^l\delta_i d_i c_i}{\sum\limits_{i=1}^l\delta_i }\leq b.$$
Denoting $\alpha_i=\frac{\delta_i}{\sum\limits_{i=1}^n\delta_i}, \ i=\overline{1,l},$ we get what we need.
It follows from the assumptions relative to matrix $C$ that for every index $1\leq i \leq l$ there exists an index $k$ such that
$$ \sum\limits_{j=1}^l c_{kj}z_j^0 \geq c_{ki}z_i^0, $$
where $c_{ki}>0.$ Hence
$$ z_i^0 \leq \frac{b_k}{c_{ki}}\leq \frac{\max\limits_{1\leq k\leq n}b_k}{\min\limits_{c_{ki}>0 }c_{ki}}=c_0<\infty.$$
Due to the arbitrariness of the solution $z_0=\{z_i^0\}_{i=1}^l$ of the system of equations (\ref{mykt17}) and inequalities (\ref{mykt18}), we obtain
$$ c(\alpha)\alpha_i d_i \leq c_0.$$
Or
$$ c(\alpha)\alpha_i \leq \frac{ c_0}{d_i}.$$
After summing over the index $i$, we get
$$ c(\alpha) \leq \sum\limits_{i=1}^l\frac{ c_0}{d_i}.$$
The boundedness of $ c(\alpha)$ is established.

Let's prove the continuity of $ c(\alpha)$.
Let us consider the functions $\sum\limits_{i=1}^l c_{ki}\alpha_i d_i, \ k=\overline{1,n}$ Every of these  function is continuous on the set $P.$ Since  $c(\alpha)$ is bounded let us denote $B=\sup\limits_{\alpha \in P}c(\alpha)<\infty.$
For sufficiently small $\varepsilon >0$
that satisfies inequalities $\frac{b_k}{\varepsilon}>B, \ k=\overline{1,n},$ let us introduce the sets $$C_k^{\varepsilon}=\left\{\alpha \in P, \ \sum\limits_{i=1}^l c_{ki}\alpha_i d_i\leq \varepsilon \right\},\quad k=\overline{1,n}.$$
If the  set $C_k^{\varepsilon}$ is nonempty
we introduce the function
 \begin{eqnarray}\label{won10}
 V_k^{\varepsilon}(\alpha)=
 \left\{\begin{array}{l l} \frac{b_k}{\sum\limits_{i= 1}^lc_{ki}\alpha_i d_i}, & \mbox{if} \quad \alpha \in P \setminus C_k^{\varepsilon},\\
\frac{b_k}{\varepsilon}, & \mbox{if} \quad \alpha \in  C_k^{\varepsilon}.
\end{array} \right.
\end{eqnarray}
If the set $C_k^{\varepsilon}$ is empty one
we  put 
$$ V_k^{\varepsilon}(\alpha)= \frac{b_k}{\sum\limits_{i= 1}^l c_{ki}\alpha_i d_i}.$$
The functions $ V_k^{\varepsilon}(\alpha)$  are continuous on the set $P$ and the equality 
$$ \min\limits_{1\leq k\leq n}V_k^{\varepsilon}(\alpha)=c(\alpha)$$
is true. Really, from the inequalities
$$ \frac{b_k}{\sum\limits_{i= 1}^l c_{ki}\alpha_i d_i}\geq   V_k^{\varepsilon}(\alpha), \quad k=\overline{1,n}, $$
it follows that
$$ \min\limits_{1\leq k\leq n}V_k^{\varepsilon}(\alpha)\leq c(\alpha).$$ 
The inverse inequality follows from the note that if for a certain point 
$\alpha \in P$ we obtain that 
$$ \min\limits_{1\leq k\leq n}V_k^{\varepsilon}(\alpha)< c(\alpha).$$ 
From the definition of $ V_k^{\varepsilon}(\alpha)$ it means that  $\min\limits_{1\leq k\leq n}V_k^{\varepsilon}(\alpha)=\frac{ b_{k_0}}{\varepsilon}>B$ for a certain $k_0.$ But this is impossible. The function $ \min\limits_{1\leq k\leq n}V_k^{\varepsilon}(\alpha)$ is continuous on $P.$

Theorem \ref{TVYA} is proved.
\end{proof}
Let us denote all non negative solutions of the set of inequalities
 \begin{eqnarray}\label{mykt15}
\sum\limits_{i=1}^l c_{ki} z_i \leq b_i,\quad k=\overline{1,n}
\end{eqnarray}
by $Z.$

\begin{prope}\label{g1}
 In the set of solutions $Z_0 = \{z_0=\{z_i^0\}_{i=1}^l=\{a(\alpha)\alpha_i d_i\}_{i=1}^l, \ \alpha \in Q\}$ of the system of equations (\ref{mykt17}) and inequalities (\ref{mykt18}) there exists a minimum of the function
$$W(\alpha)=\sum\limits_{k=1}^n [b_k- c(\alpha)\sum\limits_{i=1}^lc_{ki}\alpha_i d_i]^2.$$
This minimum is global on the set of all solutions of the system of inequalities (\ref{mykt15}),
i.e
$$\min\limits_{\alpha \in Q}W(\alpha)=\min\limits_{z \in Z} \sum\limits_{k=1}^n [b_k- \sum\limits_{i= 1}^lc_{ki}z_i]^2,$$
where $Z$ is the set of all non-negative solutions of the system of inequalities (\ref{mykt15}).
\end{prope}
\begin{proof}
The function $W(\alpha)$ is continuous on the closed bounded set $Q,$ because so is the function $c(\alpha)$ due to its continuity. According to the Weierstrass Theorem, there exists a minimum of the function
$W(\alpha).$
For any solution $z=\{z_i\}_{i=1}^l \in Z$
let's denote
$$\alpha_i=\frac{\frac{z_i}{d_i}}{\sum\limits_{j=1}^n\frac{z_j}{d_j} }\quad i=\overline{1,l}. $$
Then,
$$ Cz=\sum\limits_{i=1}^l c_iz_i=\sum\limits_{j=1}^l\frac{z_j}{d_j} \sum\limits_{i=1}^l \alpha_i d_i c_i \leq b.$$
From here,
$$ \sum\limits_{j=1}^l\frac{z_j}{d_j} \leq \min\limits_{1\leq k \leq n}\frac{b_k}{[\sum\limits_{i= 1}^l \alpha_i d_i c_i]_k}=c(\alpha).$$
Therefore,
$$ \sum\limits_{k=1}^n [b_k- \sum\limits_{i=1}^l c_{ki}z_i]^2 \geq \sum\limits_{k=1}^n [b_k- c(\alpha)\sum\limits_{i=1}^lc_{ki}\alpha_i d_i]^2\geq $$
$$\min\limits_{\alpha \in Q}\sum\limits_{k=1}^n [b_k- c(\alpha)\sum\limits_{i=1}^lc_{ki}\alpha_i d_i] ^2.$$
Taking the minimum over $z \in Z$, we have
$$\min\limits_{z \in Z} \sum\limits_{k=1}^n [b_k- \sum\limits_{i=1}^l c_{ki}z_i]^2\geq \min\limits_ {\alpha \in Q}\sum\limits_{k=1}^n [b_k- c(\alpha)\sum\limits_{i=1}^l c_{ki}\alpha_i d_i]^2.$$
The inverse inequality is obvious due to the inclusion $Z \supset Z_0, $ where $Z_0$ is the set of solutions of the system of equations (\ref{mykt17}) and inequalities (\ref{mykt18}).
Proposition \ref{g1} is proved.
\end{proof}

\begin{te}\label{TtsyVtsja5}
Let $A=||a_{ij}||_{ij=1}^n$ be a non negative nonzero matrix. The sufficient condition for the existence of a solution $p=\{p_i\}_{i=1}^n$ of the system of equations 
$$ \sum\limits_{i=1}^n a_{ki}\frac{\bar b_i p_i}{\sum\limits_{s=1}^n a_{si}p_s}=\bar b_k \quad k =\overline{1,n}, $$
in the set $P=\{p=\{p_i\}_{i=1}^n, \ p_i\geq 0,\ i=\overline{1,n},\  \sum\limits_{i=1}^n p_i=1\}$  is the existence of a solution to the system of equations
$$ z_i=\frac{\bar b_i p_i^0}{ \sum\limits_{s=1}^n a_{si}p_s^0}, \quad i=\overline{1,n}, $$
in the set $P$ for a certain non negative nonzero vector $z=\{z_i\}_{i=1}^n$ and $\bar b=Az=\{\bar b_i\}_{i=1} ^n.$
\end{te}
\begin{proof}
The proof is obvious.
\end{proof}
 \begin{te}\label{2TtsyVtsja4}
Let $A=||a_{ij}||_{ij=1}^n$ be a non negative non zero matrix and $z=\{z_i\}_{i=1}^n$ be a non negative vector such that the vector   $\bar b=Az=\{\bar b_i\}_{i=1}^n$ is strictly positive. Then there exists a solution to the set of equations
$$ z_i=\frac{\bar b_i p_i^0}{ \sum\limits_{s=1}^n a_{si}p_s^0}, \quad i=\overline{1,n}, $$ 
in the set $P=\{p=\{p_i\}_{i=1}^n,\ p_i \geq 0,\ i=\overline{1,n}, \ \sum\limits_{i=1} ^n p_i=1\}.$ 
\end{te}
\begin{proof}
Let us consider the nonlinear set of equations
\begin{eqnarray}\label{TPK1}
\frac{p_k+y_k \sum\limits_{s=1}^n a_{sk}p_s}{1+\sum\limits_{k=1}^n y_k \sum\limits_{s=1}^n a_{sk}p_s}=p_k, \quad k=\overline{1,n},
\end{eqnarray}
 where we denoted by
 $$y_k=\frac{z_k}{\sum\limits_{i=1}^n a_{ki}z_i}, \quad k=\overline{1,n}. $$
There exists a solution of this set of equations in the set $P,$ since the left part of this set of equation is a map that maps the set $P$ into itself and is continuous on it. Due to Brouwer Theorem \cite{Nirenberg} there exists a solution $p_0=\{p_i^0\}_{i=1}^n$ of the set of equations (\ref{TPK1}) in the set $P.$
From the set of equations (\ref{TPK1}) it follows that $p_0=\{p_i^0\}_{i=1}^n$
is a solution of the set of equations
\begin{eqnarray}\label{TPK2}
y_k \sum\limits_{s=1}^n a_{sk}p_s^0=\lambda p_k^0,\quad k=\overline{1,n},
\end{eqnarray}
where $\lambda=\sum\limits_{k=1}^n y_k \sum\limits_{s=1}^n a_{sk}p_s^0.$
Or,
\begin{eqnarray}\label{TPK3}
z_k \sum\limits_{s=1}^n a_{sk}p_s^0=\lambda p_k^0\sum\limits_{i=1}^n a_{ki}z_i,\quad k=\overline{1,n},
\end{eqnarray}
Summing up the left and right hand sides over $k$ from $1$ to $n$ of the equalities (\ref{TPK3}) we get
\begin{eqnarray}\label{TPK3}
\sum\limits_{s=1}^n \bar b_s p_s^0=\lambda \sum\limits_{k=1}^n p_k^0\bar b_k,\quad k=\overline{1,n}.
\end{eqnarray}
Since $\bar b_k>0,\ k=\overline{1,n},$ and  the vector 
$p_0=\{p_i^0\}_{i=1}^n$ is non zero,   we obtain
$\sum\limits_{s=1}^n \bar b_s p_s^0>0.$ From here we get $\lambda=1.$  Theorem \ref{2TtsyVtsja4} is proved.
\end{proof}

\begin{de}\label{kolja2}
Let $A=||a_{ij}||_{i,j =1}^n$ be a non negative non zero  matrix, and let  $b$ be a strictly positive vector, which does not belong to the cone formed by the column vectors of the matrix $A. $ We  say that the vector 
$z=\{z_i\}_{i=1}^n,$ being  a solution of the system of equations and inequalities
\begin{eqnarray}\label{100kolja6}
 \sum\limits_{j=1}^n a_{ij}z_j=b_i, \quad i \in I,
\end{eqnarray}
\begin{eqnarray}\label{100kolja7}
 \sum\limits_{j=1}^n a_{ij}z_j<b_i, \quad i \in J,
\end{eqnarray}
for a certain nonempty set  $I$, corresponds to the equilibrium price  vector,  which is solution of  the system of equations 
\begin{eqnarray}\label{100kolja8}
 \sum\limits_{i=1}^n a_{ki}\frac{\bar b_i p_i}{\sum\limits_{s=1}^n a_{si}p_s}=\bar b_k \quad k=\overline{1,n},
\end{eqnarray}
in the set $P=\{p=\{p_i\}_{i=1}^n,\ p_i \geq 0,\ i=\overline{1,n}, \ \sum\limits_{i=1} ^n p_i=1\},$ where $\bar b=Az=\{\bar b_i\}_{i=1}^n.$
\end{de}

 Theorem \ref{myktina19} is the basis for the determining of the equilibrium price vector in the case of partial clearing of markets.

\begin{te}\label{myktina19}
Let $A=||a_{ij}||_{i,j =1}^n$ be a non negative  non zero  matrix, and let $b=\{b_i\}_{=1}^n$ be a strictly positive vector. 
The equilibrium   price vector $p=\{p_i\}_{i=1}^n$,  being a solution  of the system of equations and inequalities
 $$  \sum\limits_{j=1}^n  a_{ij}\frac{b_j p_j}{\sum\limits_{s=1}^n a_{sj}p_s}=b_i \quad i \in I,     $$
\begin{eqnarray}\label{mari1}
   \sum\limits_{j=1}^n a_{ij}\frac{b_j p_j}{\sum\limits_{s=1}^n a_{sj}p_s}<b_i \quad i \in J,     
\end{eqnarray}
in the set $P=\{p=\{p_i\}_{i=1}^n,\ p_i \geq 0,\ i=\overline{1,n}, \ \sum\limits_{i=1} ^n p_i=1\}$
is a solution of the system of equations
\begin{eqnarray}\label{mari2}
z_i=\frac{\bar b_i p_i}{\sum\limits_{s=1}^n a_{si}p_s}, \quad i=\overline{1,n},
\end{eqnarray}
where $\bar b=\{\bar b_i\}_{i=1}^n,$ $\bar b=A z,$ the vector $z=\{z_i\}_{i=1}^n$ is determined as follows
 $z_i=z_i^0, \ i \in I, z_i=0, \ i \in J.$ The non-negative  vector $z_0^I=\{z_i^0\}_{i \in I}$ satisfies the system of equations and inequalities
\begin{eqnarray}\label{100kolja9}
   \sum\limits_{j \in I} a_{ij} z_j^0=b_i \quad i \in I,
\end{eqnarray}
\begin{eqnarray}\label{100kolja10}
  \sum\limits_{j \in I} a_{ij}z_j^0<b_i \quad i \in J.
\end{eqnarray}
 \end{te}
\begin{proof}
Let there exist a solution of the system of equations and inequalities (\ref{mari1}) with respect to the vector $p=\{p_i\}_{i=1}^n$ in the set $P.$  Due to Lemma \ref{mytinq2}, the components of the equilibrium price vector
$p=\{p_k\}_{k=1}^n$ solving the set of equations and inequalities  (\ref{mari1})  are such that $p_i=0, \ i \in J.$ 
The remaining components $p_i, \  i \in I,$ of the vector $p$ are the solution of the system of equations and inequalities
\begin{eqnarray}\label{vasja}
  \sum\limits_{j\in I} a_{ij}\frac{b_j p_j}{\sum\limits_{s\in I } a_{sj}p_s}=b_i \quad i \in I,
 \end{eqnarray}
 \begin{eqnarray}\label{vasja1}
  \sum\limits_{j\in I} a_{ij}\frac{b_j p_j}{\sum\limits_{s\in I } a_{sj}p_s}<b_i \quad i \in J.
 \end{eqnarray}
Let us  introduce the denotation
 \begin{eqnarray}\label{vasja1} 
z_j^0= \frac{b_j p_j}{\sum\limits_{s\in I } a_{sj}p_s}, \quad j \in I.
 \end{eqnarray}
It is evident that  the  equalities and inequalities 
 $$ \sum\limits_{j \in I} a_{ij} z_j^0=b_i \quad i \in I, $$
 $$ \sum\limits_{j \in I} a_{ij}z_j^0<b_i \quad i \in J, $$
 are valid.
If we introduce a vector
$z=\{z_i\}_{i=1}^n,$ where $z_i=z_i^0, \ i \in I, z_i=0, \ i \in J,$
then we obtain that the vector $z$ satisfies a system of equations and inequalities
$$ \sum\limits_{j =1}^na_{ij} z_j=b_i \quad i \in I, $$
 $$ \sum\limits_{j=1} ^na_{ij}z_j<b_i \quad i \in J. $$
The price vector
$ p=\{p_i\}_{i=1}^n,$ being a  solution of the system of equations (\ref{mari2}), is also a solution of the system equations
\begin{eqnarray}\label{mari3}
   \sum\limits_{j =1}^n a_{ij}\frac{\bar b_j p_j}{\sum\limits_{s=1}^n a_{sj}p_s}=\bar b_i \quad i=\overline{1,n},
\end{eqnarray}
due to the fact that $\bar b=Az.$ 
Theorem \ref{myktina19} is proved.
\end{proof}

\begin{te}\label{MykHon1}
Let  $A=||a_{ij}||_{i,j =1}^n$ be a  non negative non zero  matrix and let   $b=\{b_i\}_{=1}^n$ be a strictly positive vector. The necessary and sufficient condition of the existence
of the equilibrium  price vector  $p=\{p_i\}_{i=1}^n$ which is a solution of the system of equations and inequalities
  $$ \sum\limits_{j=1}^n a_{ij}\frac{b_j p_j}{\sum\limits_{s=1}^n a_{sj}p_s}=b_i \quad i \in I, $$
\begin{eqnarray}\label{MykHon2}
   \sum\limits_{j=1}^n a_{ij}\frac{b_j p_j}{\sum\limits_{s=1}^n a_{sj}p_s}<b_i \quad i \in J,
\end{eqnarray}
in the set $P=\{p=\{p_i\}_{i=1}^n,\ p_i \geq 0,\ i=\overline{1,n}, \ \sum\limits_{i=1} ^n p_i=1\}$
 is the existence of a non negative solution of the system of equations and inequalities
$$  \sum\limits_{j\in I}a_{ij} z_j=b_i \quad i \in I,     $$
\begin{eqnarray}\label{MykHon3}
   \sum\limits_{j\in I} a_{ij}z_j<b_i \quad i \in J, 
\end{eqnarray}
for a certain non empty set $I.$ 
\end{te}
\begin{proof}
Necessity.
Let there exist a solution of the system of equations and inequalities (\ref{MykHon2}) with respect to the vector $p=\{p_i\}_{i=1}^n$ in the set $P.$  Due to Lemma \ref{mytinq2},  we have that $p_i=0, \ i \in J,$ where $I$ is a non empty set.
The remaining components $p_i, \ i \in I,$ of the vector $p$ are the solution of the system of equations and inequalities
\begin{eqnarray}\label{MykHon4}
  \sum\limits_{j\in I} a_{ij}\frac{b_j p_j}{\sum\limits_{s\in I } a_{sj}p_s}=b_i \quad i \in I,
 \end{eqnarray}
 \begin{eqnarray}\label{MykHon5}
  \sum\limits_{j\in I} a_{ij}\frac{b_j p_j}{\sum\limits_{s\in I } a_{sj}p_s}<b_i \quad i \in J.
 \end{eqnarray}
Let's introduce the denotation
$$ z_i=\frac{b_i p_i}{\sum\limits_{s \in I}a_{si} p_s},\quad i \in I.$$
Then the equalities and inequalities 
 $$ \sum\limits_{j \in I} a_{ij} z_j=b_i \quad i \in I, $$
 $$ \sum\limits_{j \in I} a_{ij}z_j<b_i \quad i \in J. $$
 are valid.
It is obvious that $ z_i\geq 0, i \in I.$ The necessity is established.

Sufficiency. If there exists a non negative solution of the system of equations and inequalities (\ref{MykHon3}),  then $b_i=[A z]_i>0,\  i \in I.$  
Due  to Theorem \ref{2TtsyVtsja4},  there exists a solution $p_0^I=\{p_i^0\}_{i\in I}$  of the set of equations
$$ z_i=\frac{\bar b_i p_i^0}{ \sum\limits_{s\in I} a_{si}p_s^0}, \quad i \in I, $$
in the set $P=\{p=\{p_i\}_{i\in I}, p_i\geq 0, i \in I, \sum\limits_{i \in I}p_i=1\}$ for  non negative nonzero vector $z=\{z_i\}_{i\in I},$ where $\bar b=\{\bar b_i\}_{i\in I},$ $ \bar b_i=\sum\limits_{i \in I}a_{ij} z_j, i \in I.$

Therefore, the vector $p_0=\{p_i^0\}_{i=1}^n$ is a solution of the system of equations 
 \begin{eqnarray}\label{MykHon6}
\sum\limits_{i \in I}a_{ki} \frac{p_i^0 b_i}{\sum\limits_{s \in I}a_{si}p_s^0} =b_k, \quad k\in I,
\end{eqnarray}
and inequalities
\begin{eqnarray}\label{MykHon7}
\sum\limits_{i \in I}a_{ki} \frac{p_i^0 b_i}{\sum\limits_{s \in I}a_{si}p_s^0} < b_k, \quad k\in J.
\end{eqnarray}
 Let's construct the equilibrium vector of prices $p=\{p_i\}_{i=1}^n $ by setting $p_i=0, i \in J,$ and putting $p_i=p_i^0, i \in I.$  The price vector constructed in this way is the solution of the system of equations and inequalities (\ref{MykHon2}).
Theorem \ref{MykHon1} is proved.
\end{proof}
 The question arises for which sets $I$ there is a non-negative solution of the system of equations and inequalities (\ref{MykHon3}).
 The answer to this question is provided by Theorem \ref{TVYA}.
 
 Let us construct the matrix $C$, which appears in Theorem \ref{TVYA}. Let's put $C^I=||a_{k i}||_{k=\overline{1,n}, i \in I}.$ Then the matrix $C^I $ has dimension $n\times |I|,$ where $|I|$ is a capacity of the set $I.$
 
\begin{ce}\label{rem1}
The solution of the set of equations and inequalities (\ref{MykHon3}) exists if the matrix $C^I$ is such that  the conditions of Theorem \ref{TVYA}  are fulfilled.
\end{ce}
Let us denote $Z_0$ the set of all solutions of the set of equations and inequalities (\ref{MykHon3}) when the set $I\subset N$ runs the set of all subsets of the set $N$ for which the solution of the set of equations and inequalities (\ref{MykHon3}) exists.
For any vector $z_0 \in Z_0$, let
\begin{eqnarray}\label{TVYA1}
\bar b=\sum\limits_{i \in I} z_i a_i
\end{eqnarray}
for a certain $I \subset N$.
Then, $I=\{i, \bar b_i=b_i\}$ is a nonempty set and the vector $\bar b$
we call the vector of real consumption. In accordance with Theorem \ref{TtsyVtsja5}, \ref{2TtsyVtsja4} it corresponds to the equilibrium price vector $p_0=\{p_i^0\}_{i=1}^n,$ which is the solution of the system of equations (\ref{mari3}).

For  part of goods, the indices of which are included in the set $J=N_0\setminus I,$ the equilibrium price is $p_i^0=0, \ i \in J.$ The latter means that industries the indexes of which belongs to the set $J$ need subsidies.
  But certain funds were spent on their production, which are called the cost of these goods.
Let's introduce the generalized equilibrium price vector  by putting $p_u=\{p_i^u\}_{i=1}^n$ $p_i^u=p_i^0, \ i \in I,$ \ $ p_i^u=p_i^ c, i \in J,$ where $p_i^c=\sum\limits_{s \in I}a_{si}p_s^0$ is the cost price of the produced goods.  Each such equilibrium state  we characterize by the level of excess supply
\begin{eqnarray}\label{zaTVYA1}
R=\frac{\langle b-\bar b, p_u\rangle}{\langle b, p_u\rangle},
\end{eqnarray}
where $\langle x, y\rangle=\sum\limits_{i=1}^n x_i y_i, \ x=\{x_i\}_{i=1}^n, \ y=\{y_i\}_ {i=1}^n.$

Finding solutions of the system of equations and inequalities (\ref{mykt17}), (\ref{mykt18}) with the smallest excess supply  will require finding all possible solutions of such a system of equations and inequalities and finding among them the minimum excess supply, which can turn out to be an infeasible problem for large dimensions of the matrix $A.$ Based on Theorem \ref{TVYA} and Proposition \ref{g1} below, the solution of this problem is proposed as a quadratic programming problem.

\begin{de}\label{myktinavitka1}
Let $A$ be an indecomposable nonnegative matrix, and let $b$ be a strictly positive vector that does not belong to the cone formed by the column vectors of the matrix $A.$
The solution $z_0$ of the quadratic programming problem
\begin{eqnarray}\label{myktinavitka2}
\min\limits_{z \in Z_0}\sum\limits_{i=1}^n[b_i-\sum\limits_{k=1}^n a_{ik} z_k]^2
\end{eqnarray}
 corresponds to the real consumption vector $\bar b=A z_0\leq b.$
Assume that for the non-empty set $I=\{i, \bar b_i=b_i\}$ there exists an equilibrium price vector  $p_0=\{p_i^0\}_{i=1}^n$ which is a solution of the system equations
\begin{eqnarray}\label{20myktina21}
\sum\limits_{i=1}^n a_{ki} \frac{p_i^0\bar b_i}{\sum\limits_{s=1}^n a_{si}p_s^0}=\bar b_k, \quad k=\overline{1,n}.
\end{eqnarray}
 Then the value
\begin{eqnarray}\label{jajaja1}
  R=\frac{\langle b -\bar b, p_u  \rangle}{\langle b, p_u \rangle }
\end{eqnarray}
   will be called the generalized excess supply  corresponding to the generalized equilibrium vector of prices $p_0,$ where
$ \langle x,y \rangle=\sum\limits_{i=1}^n x_i y_i,$
$x=\{x_i\}_{i=1}^n, \ y=\{y_i\}_{i=1}^n.$
\end{de}
According to the formula (\ref{jajaja1}), the level of excess supply  for the generalized equilibrium vector is the smallest. This is the state of economic equilibrium below which the economic system cannot fall.
If this value is quite large, then the economic system may fall into a state of recession (see \cite{4Gonchar}, \cite{3Gonchar}, \cite{5Gonchar}, \cite{7Gonchar}).

Therefore, if the taxation system does not coincide with the taxation system that ensures sustainable development, then the vector $\{(1-\pi_i)x_i\}_{i=1}^n$ does not belong to the interior of the cone formed by the vectors of the columns of the matrix $A(E-A )^{-1}$ (see Theorem 16 in \cite{11Gonchar11}). This leads to the fact that in a state of economic equilibrium certain industries need subsidies for their existence.  As we can see, the reason for this is the taxation system, which leads to the fact that a certain part of the industries is unprofitable.

\section{Conclusions.}
The work describes all taxation systems that ensure sustainable development. An explicit form for such taxation systems was found and an equilibrium price vector corresponding to them was constructed.
This equilibrium price vector is determined from the condition of equality of supply and demand for resources and produced goods from them. The demand for resources and produced  goods is determined by both producers of products and consumers of finished products.
That is, the equilibrium price vector is formed both under the influence of market forces and the taxation system.
Under such an equilibrium price vector, the economy is able to function in the mode of sustainable development.
Among taxation systems that ensure sustainable development, there are perfect taxation systems that are characterized by the equality of the created gross added value in the industry of the value of the product created in the same industry.

The taxation system under which the economic system is able to function in the mode with subsidies is also completely  described.

 It is shown that in the mode of sustainable development with perfect tax system, the gross output vector in value indicators satisfies a certain system of linear homogeneous equations. As a consequence of this, the vector of the created product in value indicators coincides with the vector of created added values in the economic system.

But there are tax systems that do not coincide with tax systems that ensure sustainable development.
Under such taxation systems, only partial clearing of markets takes place. For this case, a complete description of all equilibrium states is given in which only partial clearing of the markets takes place.
For this, a complete description of all non negative  solutions of linear systems of equations and inequalities was given.
Based on this result, a complete description of all equilibrium states in which only partial clearing of the markets takes place was given.


\end{document}